\documentclass{llncs}

\usepackage{fullpage} 
\usepackage{color}
\usepackage{amsmath,amsbsy,amssymb}
\usepackage{mathtools}
\usepackage[colorlinks=true,urlcolor=webblue,linkcolor=webgreen,filecolor=webblue,citecolor=webgreen,pdfpagemode=UseOutlines,pdfstartview=FitH,pdfpagelayout=OneColumn,bookmarks=true]{hyperref}
\usepackage{doi}
\usepackage[numbers]{natbib}
\usepackage{microtype}
\usepackage{xspace}
\usepackage{tikz}

\hypersetup{pdfauthor={Gorjan Alagic, Alexander Russell}}

\definecolor{webgreen}{rgb}{0,.5,0}
\definecolor{webblue}{rgb}{0,0,.5}

\DeclareMathOperator*{\Exp}{\mathbb{E}}

\numberwithin{equation}{section}

\spnewtheorem{scheme}{Scheme}{\bfseries}{\itshape}
\spnewtheorem{assumption}{Assumption}{\bfseries}{\itshape}
\spnewtheorem*{proposition*}{Proposition}{\bfseries\upshape}{\itshape}
\spnewtheorem*{theorem*}{Theorem}{\bfseries\upshape}{\itshape}

\newcommand{\opn}{\operatorname}

\newcommand{\Z}{\mathbb{Z}}
\newcommand{\F}{\mathbb{F}}
\newcommand{\N}{\mathbb{N}}

\newcommand{\GL}{\textsf{GL}}
\newcommand{\SL}{\textsf{SL}}

\newcommand{\bits}{\{0,1\}}

\newcommand{\expref}[2]{\texorpdfstring{\hyperref[#2]{#1~\ref{#2}}}{#1~\ref{#2}}}

\newcommand{\revise}[1]{}

\newcommand{\algo}{\mathcal}
\newcommand{\negl}{\opn{negl}}
\newcommand{\poly}{\opn{poly}}
\newcommand{\inrand}{\in_R}

\newcommand{\ket}[1]{{\left\vert{#1}\right\rangle}}

\newcommand{\adv}{\textbf{Adv}}

\newcommand{\KeyGen}{\ensuremath{\mathsf{KeyGen}}}
\newcommand{\Enc}{\ensuremath{\mathsf{Enc}}}
\newcommand{\Dec}{\ensuremath{\mathsf{Dec}}}
\newcommand{\Mac}{\ensuremath{\mathsf{Mac}}}
\newcommand{\Ver}{\ensuremath{\mathsf{Ver}}}

\newcommand{\HS}{{\normalfont\textsc{Hidden Shift}}}
\newcommand{\hs}{{\textsf{HS}}}
\newcommand{\RHS}{{\normalfont\textsc{Random Hidden Shift}}}
\newcommand{\rhs}{{\textsf{RHS}}}
\newcommand{\DRHS}{{\normalfont\textsc{Decisional Random Hidden Shift}}}
\newcommand{\drhs}{{\textsf{DRHS}}}
\newcommand{\HSP}{{\normalfont\textsc{Hidden Subgroup Problem}}}
\newcommand{\hsp}{{\textsf{HSP}}}
\newcommand{\RHSP}{{\normalfont\textsc{Random Hidden Subgroup Problem}}}
\newcommand{\rhsp}{{\textsf{RHSP}}}
\newcommand{\emd}{{\textsf{EMD}}}

\title{Quantum-Secure Symmetric-Key Cryptography\\ Based on Hidden Shifts}

\author{Gorjan Alagic\inst{1} \and Alexander Russell\inst{2}}
\institute{QMATH, Department of Mathematical Sciences\\ University of
  Copenhagen\\ \email{galagic@gmail.com} \and
Department of Computer Science and Engineering\\ University of Connecticut\\ \email{acr@cse.uconn.edu}
}

\begin{document}

\maketitle

\begin{abstract}
Recent results of Kaplan et al., building on work by Kuwakado and Morii, have shown that a wide variety of classically-secure symmetric-key cryptosystems can be completely broken by \emph{quantum chosen-plaintext attacks} (qCPA). In such an attack, the quantum adversary has the ability to query the cryptographic functionality in superposition. The vulnerable cryptosystems include the Even-Mansour block cipher, the three-round Feistel network, the Encrypted-CBC-MAC, and many others.

In this article, we study simple algebraic adaptations of such schemes that replace $(\Z/2)^n$ addition with operations over alternate finite groups---such as $\Z/2^n$---and provide evidence that these adaptations are qCPA-secure. These adaptations furthermore retain the classical security properties and basic structural features enjoyed by the original schemes.

We establish security by treating the (quantum) hardness of the
well-studied \emph{Hidden Shift problem} as a cryptographic
assumption. We observe that this problem has a number of attractive
features in this cryptographic context, including random
self-reducibility, hardness amplification, and---in many cases of
interest---a reduction from the ``search version'' to the ``decisional
version.''  We then establish, under this assumption, the qCPA-security of
several such Hidden Shift adaptations of symmetric-key constructions. We show that a Hidden Shift version of the Even-Mansour block cipher
yields a quantum-secure pseudorandom function, and that a Hidden Shift
version of the Encrypted CBC-MAC yields a collision-resistant hash
function. Finally, we observe that such adaptations frustrate the direct Simon's algorithm-based attacks in more general circumstances, e.g., Feistel networks and slide attacks.
\end{abstract}

\newpage

\section{Introduction}\label{sec:intro}

The discovery of efficient quantum algorithms for algebraic problems
with longstanding roles in cryptography, like factoring and discrete
logarithm~\cite{Shor94}, has led to a systematic re-evaluation of
cryptography in the presence of quantum attacks. Such attacks can, for
example, recover private keys directly from public keys for many
public-key cryptosystems of interest. A 2010 article of \citet{KM10}
identified a new family of quantum attacks on certain generic
constructions of \emph{private-key} cryptosystems. While the attacks
rely on similar quantum algorithmic tools (that is, algorithms for the
hidden subgroup problem), they qualitatively differ in several other
respects. Perhaps most notably, they break reductions which are
information-theoretically secure\footnote{The adversary is permitted to query the oracle a polynomial number of times, but may perform arbitrarily complex computations between queries.} in the classical setting. On the
other hand, these attacks require a powerful ``quantum CPA'' setting which
permits the quantum adversary to make queries---in superposition---to
the relevant cryptosystem.

These quantum chosen-plaintext attacks (qCPA) have been generalized and expanded to apply
to a large family of classical symmetric-key constructions, including
Feistel networks, Even-Mansour ciphers, Encrypted-CBC-MACs, tweakable block ciphers, and others~\cite{KLLN16, KM10, KM12, SS17}. A unifying feature of all these new attacks, however,
is an application of Simon's algorithm for recovering ``hidden
shifts'' in the group $(\Z/2)^n$. Specifically, the attacks exploit an
internal application of addition $(\operatorname{mod} 2)$ to construct
an instance of a hidden shift problem---solving the hidden shift
problem then breaks the cryptographic construction. As an illustrative
example, consider two (independent) uniformly random permutations
$P, Q: \{0,1\}^n \rightarrow \{0,1\}^n$ and a uniformly random element
$z$ of $\{0,1\}^n$. It is easy to see that no classical algorithm can
distinguish the function $(x,y) \mapsto (P(x), Q(y))$ from the
function $(x,y) \mapsto (P(x), P(y \oplus z))$ with a polynomial
number of queries; this observation directly motivates the classical
Even-Mansour block-cipher construction. On the other hand, an
efficient quantum algorithm with oracle access to
$(x,y) \mapsto (P(x), P(y \oplus z))$ can apply Simon's algorithm to
recover the ``hidden shift'' $z$ efficiently; this clearly allows the algorithm
to distinguish the two cases above.

While these attacks threaten many classical private-key constructions,
they depend on an apparent peculiarity of the group $(\Z/2)^n$---the
\HS~problem over $(\Z/2)^n$ admits an efficient
quantum algorithm. In contrast, \HS~problems in general have
resisted over 20 years of persistent attention from the quantum
algorithms community. Indeed, aside from Simon's polynomial-time
algorithm for hidden shifts over $(\Z/2)^n$, generalizations to certain groups of constant exponent~\cite{FIMSS14}, and Kuperberg's
$2^{O(\sqrt{\log N})}$ algorithm for hidden shifts over $\Z/N$~\cite{Kuperberg05}, very little is known. 
This dearth of progress is not for lack of motivation. In fact, it is well-known that efficient quantum algorithms for \textsc{Hidden Shift} over $\Z/N$ would (via a well-known reduction from the {\HSP} on $D_N$) yield efficient quantum attacks
on important public-key cryptosystems~\cite{Regev04-lattice,Regev05}, including prime candidates for quantum security and the eventual replacement of RSA in Internet cryptography~\cite{NIST16}. Likewise, efficient algorithms for the symmetric group would yield polynomial-time quantum algorithms for \textsc{Graph Isomorphism}, a longstanding challenge in the area.

On the other hand, $(\Z/2)^n$ group structure is rather incidental to
the security of typical symmetric-key constructions. For example, the classical Even-Mansour construction defines a block cipher $E_{k_1, k_2}(m)$ by the rule
\[
E_{k_1, k_2}(m) = P(m \oplus k_1) \oplus k_2\,,
\]
where $P$ is a public random permutation and the secret key
$(k_1, k_2)$ is given by a pair of independent elements drawn
uniformly from $(\Z/2)^n$. The security proofs, however, make no
particular assumptions about group structure, and apply if the
$\oplus$ operation is replaced with an alternative group operation,
e.g., $+$ modulo $N$ or multiplication in $\F_{2^n}$.

This state of affairs suggests the possibility of ruling out quantum
attacks by the simple expedient of adapting the underlying group in
the construction. Moreover, the apparently singular features of
$(\Z/2)^n$ in the quantum setting suggest that quite mild adaptations
may be sufficient. As mentioned above, many classical security proofs
are unaffected by this substitution; our primary goal is to add
security against quantum adversaries. 
 \emph{Our approach is to
reduce well-studied Hidden Shift problems directly to the
security of these symmetric-key cryptosystems. Thus, efficient quantum chosen-plaintext
attacks on these systems would resolve long-standing open questions in
quantum complexity theory.}

\subsection{Contributions}

\subsubsection{Hidden Shift as a cryptographic primitive.}

We propose the intractability of the \textsc{Hidden Shift} problem as a fundamental assumption for establishing quantum security of cryptographic schemes. In the general problem, we are given two functions on some finite group $G$, and a promise that one is a shift of the other; our task is to identify the shift. Our assumptions have the following form: 


\begin{assumption}[The $\mathcal{G}$-Hidden Shift Assumption,
  informal]\label{assumption-one}\label{core-assumption-inf} Let
  $\mathcal{G} = \{ G_i \mid i \in I\}$ be a family of finite groups
  indexed by a set $I \subset \{0,1\}^*$. For all polynomial-time
  quantum algorithms $\algo A$,
  $$
  \Exp_f \Bigl[ \min_{\substack{s \in G_i}} \Pr\bigl[\algo A^{f,
    f_s}(i) = s\bigr]\Bigr] \leq \negl(|i|)\,,
  $$
  where $f_s(x) = f(sx)$, the expectation is taken over random choice
  of the function $f$, the minimum is taken over all shifts
  $s \in G_i$, and the probability is taken over internal randomness
  and measurements of $\mathcal{A}$.
\end{assumption}

This assumption asserts that there is no quantum algorithm for
\HS~(over $\mathcal G$) in the \emph{worst-case} over $s$, when
function values are chosen randomly. Note that the typical formulation
in the quantum computing literature is worst case over $s$ and $f$; on
the other hand, known algorithmic approaches are invariant under
arbitrary relabeling of the value space of $f$. The ``random-valued''
case thus seems satisfactory for our cryptographic purposes. (In fact,
our results can alternatively depend on the case where $f$ is
injective, rather than random.) See
\expref{Section}{sec:HSP-assumption} below for further discussion and
precise versions of \expref{Assumption}{assumption-one}.  In general,
formulating such an assumption requires attention to the encoding of
the group. However, we will focus entirely on groups with
conventional encodings which directly provide for efficient group
operations, inversion, generation of random elements,
etc. Specifically, we focus on the two following particular variants:
\begin{assumption}[The $2^n$-Cyclic Hidden Shift
  Assumption]\label{cyclic-assumption}
  This is the Hidden Shift Assumption with the group family
  $\mathcal{C}_2 = \{ \Z/2^n \mid n \geq 0\}$ where the index consists of
  the number $n$ written in unary.
\end{assumption}
\begin{assumption}[The Symmetric Hidden Shift
  Assumption]\label{symmetric-assumption}
  This is the Hidden Shift Assumption with the group family
  $\mathcal{S} = \{ S_n \mid n \geq 0\}$ where $S_n$ denotes the
  symmetric group on $n$ symbols and the index consists of the number
  $n$ written in unary.
\end{assumption}
\noindent
In both cases the size of the group is exponential in the length of the index.

We remark that the \textsc{Hidden Shift} problem has polynomial
quantum query complexity~\cite{EHK04}---thus one cannot hope that
\textsc{Hidden-Shift}-based schemes possess information-theoretic
security in the quantum setting (as they do in the classical setting);
this motivates introduction of \HS~intractability assumptions.

\begin{sloppypar}
To explore the hardness of \HS~problems against quantum polynomial-time (QPT) algorithms, we describe several reductions. First, we prove that \HS~is equivalent to a randomized version of the problem where the shift $s$ is random (\RHS), and provide an amplification theorem which is useful in establishing security of schemes based on \expref{Assumption}{core-assumption-inf}. 
\end{sloppypar}

\begin{proposition*}[Amplification, informal]
  Assume there exists a QPT algorithm which solves \RHS~for an
  inverse-polynomial fraction of inputs. Then there exists a QPT
  algorithm for solving both \HS~and \RHS~for all but a negligible
  fraction of inputs.
\end{proposition*}

We then show that, for many group families, \textsc{Hidden Shift} over the relevant groups is equivalent to a \emph{decisional} version of the problem. In the decisional version, we are guaranteed that the two functions are either (i.) both random and independent, or (ii.) one is random and the other is a shift; the goal is to decide which is the case.

\begin{theorem*}[Search and decision are equivalent, informal]
  Let $\mathcal G$ be the group family $\mathcal{C}_2$ or the group
  family $\mathcal S$ (or a group family with an efficient subgroup
  series). Then there exists a QPT algorithm for \RHS~(with at most inverse-poly
  error) over $\mathcal G$ if and only if there exists a QPT algorithm for \DRHS~(with at most
  inverse-poly error) over $\mathcal G$.
\end{theorem*}

Finally, we provide some evidence that \HS~over the family $\mathcal{C}_2$ is as hard as \HS~over general cyclic groups. Specifically, we show that efficient algorithms for an approximate version of \HS~over $\mathcal{C}_2$ give rise to efficient algorithms for the same problem over $\mathcal{C}$, the family of all cyclic groups.

We also briefly discuss the connections between \HS, the assumptions above, and assumptions underlying certain candidates for quantum-secure public-key cryptography~\cite{DMR15, Regev04-lattice}. For completeness, we recall known connections to the \HSP.
Both the \HS~and \HSP~families have received significant attention from the quantum algorithms community, and are believed to be quantumly hard with the exception of particular families of groups~\cite{DMR15, HMRRS10,MRS07,MRS08,Regev04-lattice}. 

\subsubsection{Quantum-secure symmetric-key cryptographic schemes.}

With the above results in hand, we describe a generic method for using \expref{Assumption}{core-assumption-inf} to ``adapt'' classically-secure schemes in order to remove vulnerabilities to quantum chosen-plaintext attacks. The adaptation is simple: replace the underlying $(\Z/2)^n$ structure of the scheme with that of either $\mathcal C_2$  or $\mathcal S$. This amounts to replacing bitwise XOR with a new group operation. In the case of $\mathcal C_2$, the adaptation is particularly simple and efficient.  

While our basic approach presumably applies in broad generality, we focus on three emblematic examples: the Even-Mansour construction---both as a PRF and as a block cipher---and the CBC-MAC construction. We focus throughout on the group families $\mathcal{C}_2$ and $\mathcal{S}$, though we also discuss some potential advantages of other choices (see \expref{Section}{sec:groups}). Finally, we discuss related quantum attacks on cryptographic constructions, including the 3-round Feistel cipher and 
quantum slide attacks~\cite{KLLN16}. We remark that the Feistel cipher over groups other than $(\Z/2)^n$ has been considered before, in a purely classical setting~\cite{PRS03}.

\paragraph{Hidden Shift Even-Mansour.}

Following the prescription above, we define group variants of the
Even-Mansour cipher. We give a reduction from the worst-case
\HS~problem to the natural \emph{distinguishability} problem (i.e.,
distinguishing an Even-Mansour cipher from a random permutation).
Thus, under the Hidden Shift Assumption, the Even-Mansour construction
is a quantum-query-secure pseudorandom function (qPRF). In particular,
key-recovery is computationally infeasible, even for a quantum
adversary. We also provide (weaker) reductions between \HS~and the
problem of breaking Even-Mansour in the more challenging case where
the adversary is provided access to both the public permutation and
its inverse (and likewise for the encryption map).  In any case, these
adaptations frustrate the ``Simon algorithm key recovery
attack''~\cite{KM12, KLLN16}, as this would now require a subroutine
for \HS~in the relevant group family. Moreover, one can also apply
standard results (see, e.g.,~\cite{HN04}) to show that, over some
groups, all bits of the key are as hard as the entire key (and hence,
by our reductions, as hard as~\HS). We remark that considering
$\Z/{2^n}$ structure to define an adaptation of Even-Mansour has been
considered before in the context of classical slide
attacks~\cite{DKS15}.

\paragraph{Hidden Shift CBC-MAC.}

Following our generic method for transforming schemes, we define group
variants of the Encrypted-CBC-MAC. We establish that this primitive is
collision-free against quantum adversaries. Specifically, we show that
any efficient quantum algorithm which discovers collisions in the
Hidden-Shift Encrypted-CBC-MAC with non-negligible probability would
yield an efficient worst-case quantum algorithm for \HS~over the
relevant group family. As with Even-Mansour, this adaptation also
immediately frustrates the Simon's algorithm collision-finding
attacks~\cite{KLLN16, SS17}.

\paragraph{Feistel ciphers, slide attacks.}

We also define group variants of the well-known Feistel cipher for constructing pseudorandom permutations from pseudorandom functions. Our group variants frustrate Simon-style attacks~\cite{KM10}; a subroutine for the more general \HS~problem is now required. Finally, we also address the exponential quantum speedup of certain classical slide attacks, as described in~\cite{KLLN16}. We show how one can once again use \HS~to secure schemes vulnerable to these ``quantum slide attacks.''

\section{Preliminaries}\label{sec:prelims}

\paragraph{Notation; remarks on finite groups.}

For a finite group $G$ and an element $s \in G$, let
$L_s: G \rightarrow G$ denote the permutation given by left
multiplication by $s$, so $L_s: x \mapsto s \cdot x$. We discuss a
number of constructions in the paper requiring computation in finite
groups and assume, throughout, that elements of the group in question
have an encoding that efficiently permits such natural operations as
product, inverse, selection of uniformly random group elements,
etc. As our discussion focuses either on specific groups---such as
$(\Z/2)^n$ or $\Z/N$---where such encoding issues are straightforward
or, alternatively, \emph{generic} groups in which we assume such
features by fiat, we routinely ignore these issues of encoding.

\paragraph{Classical and quantum algorithms.}

Throughout we use the abbreviation PPT for ``probabilistic polynomial
time,'' referring to an efficient classical algorithm, and QPT for
``quantum polynomial time,'' referring to an efficient quantum
algorithm. Our convention is to denote algorithms of either kind with
calligraphic letters, e.g., $\algo A$ will typically denote an
algorithm which models an adversary. If
$f$ is a function,
the notation $\algo A^f$ stands for an algorithm (either classical or
quantum) with oracle access to the function $f$. A classical oracle is
simply the black-box gate $x \mapsto f(x)$; a quantum oracle is the
unitary black-box gate
$\ket{x}\ket{y} \mapsto \ket{x}\ket{y \oplus f(x)}$. Unless stated
otherwise, oracle QPT algorithms are assumed to have quantum oracle
access.

\paragraph{Quantum-secure pseudorandomness.}

We now set down a way of quantifying the ability of a QPT adversary to
distinguish between families of functions. Fix a function family
$\mathcal F \subset \{ h : \bits^m \rightarrow \bits^\ell\}$, a
function $f : \bits^n \times \bits^m \rightarrow \bits^\ell$, and
define $f_k := f(k, \cdot)$. We say that $f$ is an \emph{indexed
  subfamily} of $\mathcal F$ if $f_k \in \mathcal F$ for every
$k \in \bits^n$. We will generally assume that $m$ and $\ell$ are
polynomial functions of $n$ and treat $n$ to be the complexity (or
security) parameter.

\begin{definition} Let $\mathcal F$ be a function family, $f$ an indexed subfamily, and $\algo D$ an oracle QPT algorithm. The distinguishing advantage of $\algo D$ is the quantity
$$
\emph{\adv}^{\algo D}_{\mathcal F, f} := \left|\underset{k \inrand \bits^n}{\Pr}\Bigl[\mathcal D^{f_k}(1^n) = 1\Bigr] - 
\underset{g \inrand \mathcal F}{\Pr}\Bigl[\mathcal D^g(1^n) = 1\Bigr] \right |\,.
$$
\end{definition}

Next, we define efficient indexed function families which are pseudorandom against QPT adversaries. We emphasize that these function families are computed by deterministic classical algorithms.

\begin{definition}\label{def:qPRF}
Let $\mathcal{F}_n$ be the family of all functions from $m(n)$ bits to $\ell(n)$ bits, and $f$ a efficiently computable, indexed subfamily of $\bigcup_n \mathcal{F}_n$ (so that $f_k \in \mathcal{F}_n$ for $|k| = n$). We say that $f$ is a quantum-secure pseudorandom function $(qPRF)$ if $\emph{\adv}^{\algo D}_{\mathcal{F}_n, f} \leq \negl(n)$ for all QPT $\algo D$.
\end{definition}

It is known how to construct qPRFs from standard assumptions (i.e., existence of quantum-secure one-way functions)~\cite{Zhandry12}.

The pseudorandom function property is not enough in certain applications, e.g., in constructing block ciphers. It is then often useful to add the property that each function in the family is a permutation, which can be inverted efficiently (provided the index is known).

\begin{definition}\label{def:qPRP}
 Let $\mathcal P$ be the family of all permutations, and $f$ an efficiently computable, indexed subfamily of $\mathcal P$. We say that $f$ is a quantum-secure pseudorandom permutation (qPRP) if (i.) $f$ is a qPRF, (ii.) each $f_k$ is a permutation, and (iii.) there is an efficient algorithm which, given $k$, computes the inverse $f_k^{-1}$ of $f_k$.
\end{definition}

A recent result shows how to construct qPRPs from one-way functions~\cite{Zhandry16}. Finding simpler constructions is an open problem. Two simple constructions which are known to work classically, Even-Mansour and the 3-round Feistel, are both broken by a simple attack based on Simon's algorithm for \HS~on $(\Z/2)^n$. As we discuss in detail later, we conjecture that the adaptations of these constructions to other group families are qPRPs.


We will also make frequent use of a result of Zhandry (Theorem 3.1 in~\cite{Zhandry12a}) which states that $2k$-wise independent functions are indistinguishable from random to quantum adversaries making no more than $k$ queries. 
\begin{theorem}\label{thm:k-wise}
Let $\mathcal H$ be a $2k$-wise independent family of functions with domain $\mathcal X$ and range $\mathcal Y$. Let $\algo D$ be a quantum algorithm making no more than $k$ oracle queries. Then
$$
\underset{h \inrand \mathcal H}{\Pr}\Bigl[\mathcal D^{h}(1^n) = 1\Bigr] = \underset{g \inrand \mathcal Y^\mathcal X}{\Pr}\Bigl[\mathcal D^g(1^n) = 1\Bigr] \,.
$$
\end{theorem}


\paragraph{Collision-freeness.}

We will also need a (standard) definition of collision-resistance against efficient quantum adversaries with oracle access. 

\begin{definition} Let $f : \bits^* \times \bits^* \rightarrow \bits^{*}$ be an efficiently-computable function family defined for all $(k,x)$ for which $|x| = m(|k|)$ (for a polynomial $m$). We say that $f$ is collision-resistant if for all QPT $\algo A$, 
$$
\Pr_{k \inrand \bits^n} \Bigl[ \algo A^{f_k}(1^n) = (x, y) \wedge f_k(x) = f_k(y) \wedge x \neq y \Bigr] \leq \negl(n)\,.
$$
\end{definition}

\section{Hidden Shift as a cryptographic primitive}
\label{sec:HSP-assumption}

We begin by discussing a few versions of the basic oracle promise problem related to finding hidden
shifts of functions on groups. In the problems below, the relevant
functions are given to the algorithm via black-box oracle access and
we are interested in the setting where the complexity of the algorithm
(both number of queries and running time) scales in $\poly(\log |G|)$. 

\subsection{Hidden Shift problems}\label{sec:hidden-shift}

\subsubsection{Basic definitions.}

We begin with the \HS~problem. As traditionally formulated in the quantum computing literature, the problem is the following:
\begin{problem}[The traditional Hidden Shift problem] \label{problem:HS-trad}
Let $G$ be a group and $V$ a set. Given oracle access to an injective function $f: G \rightarrow V$ and an unknown shift $g = f \circ L_s$ of $f$, find $s$.
\end{problem}

It is convenient for us to parameterize this definition in terms of a
specific group family and fix the range of the oracles $f$ and $g$. This yields our basic asymptotic definition for the problem.
\begin{problem}[$\HS$~($\hs$)] \label{problem:HS} Let
  $\mathcal{G} = \{ G_i \mid i \in I\}$ be a family of groups with
  index set $I \subset \{0,1\}^*$ and let $\ell: \N \rightarrow \N$ be
  a polynomial. Then the \HS\ problem over $\mathcal{G}$ (with length
  parameter $\ell$) is the following: given an index $i$ and oracle
  access to a pair of functions
  $f, g: G_i \rightarrow \{0,1\}^{\ell(|i|)}$ where $g(x) = f(sx)$,
  determine $s \in G_i$.  We assume, throughout, that
  $2^{\ell(|i|)} \gg |G_i|$.
\end{problem}

This generic formulation is more precise, but technically still
awkward for cryptographic purposes as it permits oracle access to
completely arbitrary functions
$f$. 
To avoid this technical irritation, we focus on the performance of
\HS\ algorithms over specific classes of functions $f$. Specifically,
we either assume $f$ is random or that it is injective. When a
\HS~algorithm is applied to solve problems in a typical computational
setting, the actual functions $f, g$ are injective and given by
efficient computations. We remark that established algorithmic
practice in this area ignores the actual function values altogether,
merely relying on the structure of the level sets of the function
\[
\Phi(x,b) = \begin{cases} f(x) & \text{if $b=0$,}\\
  g(x) & \text{if $b=1$.} \end{cases}
\]
In particular, such structural conditions of $f$ appear to be irrelevant to the success of current quantum-algorithmic techniques for the problem. This motivates the following notion of ``success'' for an algorithm.

\begin{sloppypar}
\begin{definition}[Completeness]\label{def:completeness} Let $\algo A$ be an algorithm for the
  \HS~problem on $\mathcal{G}$ with length parameter $\ell$.  Let $f$
  be a function defined on all pairs $(i,x)$ where $x \in G_i$ so that
  $f(i,x) \in \{0,1\}^{\ell(|i|)}$. Then we define the
  \emph{completeness of $\algo A$ relative to $f$} to be the quantity
  \[
    1 - \epsilon_f(i) \triangleq \min_{s \in G_i} \Pr[{\algo A}^{f,f_s}(i) = s]\,.
  \]
  The \emph{completeness of $\algo A$ relative to random functions} is the average
  \[
    1 - \epsilon_R(i) \triangleq \Exp_{f}\left[ \min_{s \in G_i} \Pr[{\algo A}^{f,f_s}(i) = s]\right] = \Exp_f [1 - \epsilon_f(i)]\,,
  \]
  where $f(i,x)$ is drawn uniformly at
  random. Note that these notions are worst-case in $s$,
  the shift.
\end{definition}
\end{sloppypar}

Note that this definition does not specify how the algorithm should behave on instances that are
not hidden shifts. For simplicity, we assume that the algorithm returns a value for $s$ in any case, with no particular guarantee on
$s$ in the case when the functions are not shifts of each other.


Our basic hardness assumption is the following:
\begin{assumption}[The $\mathcal{G}$-Hidden Shift Assumption; randomized]\label{core-assumption} Let
  $\mathcal{G} = \{ G_i \mid i \in I\}$ be a family of finite groups
  indexed by a set $I \subset \{0,1\}^*$ and $\ell: \N \rightarrow \N$
  be a length parameter. Then for all efficient algorithms $\algo A$,
  $1 - \epsilon_R(i) = \negl(|i|)$.
\end{assumption}
For completeness, we also record a version of the assumption for injective $f$.
In practice, our cryptographic constructions will rely only on the randomized version.

\begin{assumption}[The $\mathcal{G}$-Hidden Shift Assumption; injective] Let
  $\mathcal{G} = \{ G_i \mid i \in I\}$ be a family of finite groups
  indexed by a set $I \subset \{0,1\}^*$ and $\ell: \N \rightarrow \N$
  be a length parameter. Then for all efficient algorithms $\algo A$
  there exists an injective $f$ (satisfying the criteria of
  \expref{Definition}{def:completeness} above), so that
  $1 - \epsilon_f(i) = \negl(|i|)$.
\end{assumption}

In preparation for establishing results on security amplification, we define two additional variants of the \HS~problem: a variant where both the function and the shift are randomized, and a decisional variant. Our general approach for constructing security proofs will be to reduce one of these variants to the problem of breaking the relevant cryptographic scheme. As we will later show, an efficient solution to either variant implies an efficient solution to both, which in turn results in a violation of \expref{Assumption}{core-assumption} above.

\begin{problem}[\RHS~(\rhs)] \label{problem:RHS} Let
  $\mathcal{G} = \{ G_i \mid i \in I\}$ be a family of finite groups
  indexed by a set $I \subset \{0,1\}^*$ and $\ell: \N \rightarrow \N$
  be a length parameter. Then the \RHS\ problem over $\mathcal G$ is
  the \HS\ problem where the input function
  $f(i,x)$ is drawn uniformly and the
  shift $s$ is drawn (independently and uniformly) from $G_i$.
\end{problem}

We define the completeness $1 - \epsilon(i)$ for a \RHS~algorithm $\algo A$ analogously to \expref{Definition}{def:completeness}. Observe that a small error is unavoidable for any algorithm, as there exist pairs of functions for which $s$ is not uniquely defined. 
We will also need a decisional version of the problem, defined as follows.

\begin{problem}[\DRHS~(\drhs)] \label{problem:DRHS} Let
  $\mathcal{G} = \{ G_i \mid i \in I\}$ be a family of finite groups
  indexed by a set $I \subset \{0,1\}^*$ and $\ell: \N \rightarrow \N$
  be a length parameter. The
  \DRHS\ problem is the following: Given $i$ and oracle access to two
  functions $f, g: G_i \rightarrow \{0,1\}^{\ell(|i|)}$ with the promise
  that either (i.)  both $f$ and $g$ are drawn independently at
  random, or (ii.)  $f$ is random and $g = f \circ L_s$ for some
  $s \in G$, decide which is the case.
\end{problem}

We say that an algorithm for \drhs~has completeness $1-\epsilon(i)$
and soundness $\delta(i)$ if the algorithm errs with probability no
more than $\epsilon(i)$ in the case that the functions are shifts and
errs with probability no more than $\delta(i)$ in the case that the
functions are drawn independently.

Next, we briefly recall the definition of the (closely-related) \HSP\@. The problem is primarily relevant in our context because of its historical significance (and relationship to \HS); we will not use it directly in any security reductions.
\begin{problem}[\HSP~(\hsp)]
Let $G$ be a group and $S$ a set. Given a function $f : G \rightarrow S$, and a promise that there exists $H \leq G$ such that $f$ is constant and distinct on the right cosets of $H$, output a complete set of generators for $H$.
\end{problem}
Some further details, including explicit reductions between \hs~and \hsp, are given in \expref{Appendix}{appendix:HSP}. 

Of interest are both classical and quantum algorithms for solving the various versions of \hs~and \hsp. The relevant metrics for such algorithms are the query complexity (i.e., the number of times that the functions are queried, classically or quantumly) as well as their time and space complexity. An algorithm is said to be efficient if all three are polynomial in $\log |G|$.

\subsubsection{Hardness results.}

Next, we establish several reductions between these problems. Roughly, these results show that the average-case and decisional versions of the problem are as hard as the worst-case version.

\paragraph{Self-reducibility and amplification.}

First, we show that (i.) both \hs~and \rhs~are random self-reducible, and (ii.) an efficient solution to \rhs~implies an efficient solution to \hs. 
\begin{proposition}\label{prop:RHS-HS-amplify}
  Let $\mathcal{G} = \{ G_i \mid i \in I\}$ be a family of finite
  groups indexed by a set $I \subset \{0,1\}^*$ and
  $\ell: \N \rightarrow \N$ be a length parameter.  Assume there
  exists a QPT $\algo A$ which solves \textsc{Random Hidden Shift}
  over $\mathcal{G}$ (with parameter $\ell(|i|)$) with
  inverse-polynomial completeness. Then there exists a QPT $\algo A'$
  which satisfies all of the following:
  \begin{enumerate}
  \item $\algo A'$ solves \HS~with random $f$ with completeness $1 - \negl(|i|)$;
  \item $\algo A'$ solves \HS~for any injective $f$ with completeness $1 - \negl(|i|)$;
  \item $\algo A'$ solves \RHS~with completeness $1 - \negl(|i|)$.
  \end{enumerate}
\end{proposition}
\begin{proof}
  We are given oracles $f, g$ and a promise that $g = f \circ
  L_s$. For a particular choice of $n$, there is an explicit
  (polynomial-size) bound $k$ on the running time of $\algo A$. Let
  $\mathcal H$ be a $2k$-wise independent function family which maps
  the range of $f$ to itself. The algorithm $\algo A'$ will repeatedly
  execute the following subroutine. First, an element
  $h \in \mathcal H$ and an element $t \in G_i$ are selected
  independently and uniformly at random. Then $\algo A$ is executed
  with oracles
$$
f' := h \circ f 
\qquad \text{and} \qquad
g' := h \circ g \circ L_t\,.
$$
It's easy to see that $g' = f' \circ L_{st}$. If $\algo A$ outputs a group element $r$, $\algo A'$ checks if $g'(x) = f'(rx)$ at a polynomial number of random values $x$. If the check succeeds, $\algo A'$ outputs $rt^{-1}$ and terminates. If the check fails (or if $\algo A$ outputs garbage), we say that the subroutine fails. The subroutine is repeated $m$ times, each time with a fresh $h$ and $t$. 

Continuing with our fixed choice of $f$ and $g$, we now argue that $\algo A$ (when used as above) cannot distinguish between $(f', g')$ and the case where $f'$ is uniformly random, and $g'$ is a uniformly random shift of $f'$. First, the fact that the shift is randomized is clear. Second, if $f$ is injective, then $f'$ is simply $h$ with permuted inputs, and is thus indistinguishable from random (by the $2k$-wise independence of $h$ and \expref{Theorem}{thm:k-wise}). Third, if $f$ is random, then it is indistinguishable from injective (by the collision bound of~\cite{Zhandry15}), and we may thus apply the same argument as in the injective case.


It now follows that, with inverse-polynomial probability $\epsilon$ (over the choice of $h$ and $t$), the instance $(f', g')$ is indistinguishable from an instance $(\varphi, \varphi_{st})$ on which the subroutine succeeds with inverse-polynomial probability $\delta$. After $m$ repetitions of the subroutine, $\algo A'$ will correctly compute the shift $r = st$ with probability at least $(1 - \epsilon \delta)^m \approx e^{-\epsilon \delta m}$, as desired.
\qed
\end{proof}

\paragraph{Decision versus search.}

Next, we consider the relationship between searching for shifts (given the promise that one exists,) and deciding if a shift exists or not. Roughly speaking, we establish that the two problems are equivalent for most group families of interest. We begin with a straightforward reduction from \drhs~to \rhs. 

\begin{proposition}
  If there exists a QPT algorithm for \RHS~on $\mathcal{G}$ with
  completeness $1 - \epsilon(i)$, then there exists a QPT algorithm for \DRHS~on
  $\mathcal{G}$ with completeness $1 - \epsilon(i)$, and negligible
  soundness error.
\end{proposition}
\begin{proof}
Let $\ell(\cdot)$ be the relevant length parameter. Consider an \rhs~algorithm for $G$ with completeness $1 - \epsilon$ and the following adaptation to \drhs.
\begin{itemize}
  \item Run the \rhs~algorithm.
  \item When the algorithm returns a purported shift $s$, check $s$
    for veracity with a polynomial number of (classical) oracle
    queries to $f$ and $g$ (ensuring that $g(x_i) = f(sx_i)$ for
    $k(n)$ distinct samples $x_1, \ldots, x_{k(n)}$).
\end{itemize}

Observe that if $f$ and $g$ are indeed hidden shifts, this procedure
will determine that with probability $1 - \epsilon$. When $f$ and $g$
are unrelated random functions, the ``testing'' portion of the
algorithm will erroneously succeed with probability no more than
$|G| \cdot 2^{-k \cdot \ell(|i|)}$. Thus, under the assumption that $|G| \geq k(n)$, the resulting \drhs~algorithm has completeness $1 - \epsilon$ and soundness $|G| \cdot 2^{-k \cdot \ell(|i|)}$. For any nontrivial length function $\ell$, this soundness can be driven exponentially close to zero by choosing $k = \log |G| + k'$. 
\qed
\end{proof}

On the other hand, we are only aware of reductions from \rhs~to \drhs~under the additional assumption that $G$ has a ``dense'' tower of subgroups. In that case, an algorithmic approach of \citet{FZ13} can be adapted to provide a reduction. Both $S_n$ and $\Z/2^n$ have such towers.

\begin{proposition}\label{prop:search-to-decision}
Let $\mathcal G$ be either the group family $\{ \Z/{2^n}\}$, or the group family $\{S_n\}$. If there exists a QPT algorithm for \DRHS~on $\mathcal G$ with at most inverse-polynomial completeness and soundness errors, then there exists a QPT algorithm for \RHS~on $\mathcal G$ with negligible completeness error.
\end{proposition}
\begin{proof}
  The proof adapts techniques of \cite{FZ13} to our
  probabilistic setting, and relies on the fact that these group
  families have an efficient subgroup tower. Specifically, each $G_i$ possesses
  a subgroup series
  $\{1\} = G^{(0)} < G^{(1)} < G^{(2)} < \cdots < G^{(s)} = G_i$ for which (i.)
  uniformly random sampling and membership in $G^{(t)}$ can be
  performed efficiently for all $t$, and (ii.) for all $t$, there is
  an efficient algorithm for producing a left transversal of
  $G^{(t-1)}$ in $G^{(t)}$. For $\Z/{2^n}$, the subgroup series is
  $\{1\} < \Z/2 < \Z/{2^2} < \Z/{2^3} < \cdots$.
For $S_n$ (i.e., the group of permutations of $n$ letters), the subgroup series is $\{1\} < S_1 < S_2 < S_3 < \cdots$, where each step of the series adds a new letter. We remark that such series can be efficiently computed for general permutation groups using a \emph{strong generating set}, which can be efficiently computed from a presentation of the group in terms of generating permutations~\cite{FHL80}.

We recursively define a \rhs~algorithm by considering the case of a group $G$ with a subgroup $H$ of polynomial index with a known left
transversal $A = \{a_1, \ldots, a_k\}$ (so that $G$ is the disjoint
union of the $a_i H$). Assume that the \drhs~algorithm for $H$ has soundness $\delta_H$ and completeness $1-\epsilon_H$. In this case, the algorithm (for $G$) may proceed as
follows:

\begin{enumerate}
\item For each $\alpha \in A$, run the \drhs~algorithm on the two functions $f$ and $\check{g}: x \mapsto g(\alpha x)$ restricted to the subgroup $H$.
\item If exactly one of these recursive calls reports that the function $f$ and $x \mapsto g(\alpha x)$ are hidden shifts, recursively apply the \rhs~algorithm to recover the hidden shift $s'$ (so that $f(x) = g(\alpha s' x)$ for $x \in H$). Return the shift $s = \alpha s'$.
\item Otherwise assert that the functions are unrelated random functions.
\end{enumerate}
In the case that $f$ and $g$ are independent random functions, the algorithm above errs with probability no more than $[G:H] \delta_{H}$. 

Consider instead the case that $f: G \rightarrow S$ is a random
function and $g(x) = f(sx)$ for an element $s \in G$. Observe that if
$s^{-1} \in \alpha_i H$, so that $s^{-1} = \alpha_i h_s$ for an
element $h_s \in H$, we have $g(\alpha_i h_s x) = f(x)$. It follows
that $f$ and $\check{g}: x \mapsto g(\alpha_i x)$ are shifts of each
other; in particular, this is true when restricted to the subgroup
$H$. Moreover, the hidden shift $s$ can be determined directly from
the hidden shift between $f$ and $\check{g}$. Note that, as above, the
probability that any of the recursive calls to \drhs~are answered
incorrectly is no more than $[G:H] \delta_{H} + \epsilon_{H}$. 

It remains to analyze the completeness of the resulting recursive \rhs~algorithm: in the case of the subgroup chain above, let $\gamma_t$ denote the completeness of the resulting \rhs~algorithm on $G^{(t+1)}$ and note that
\[
\gamma_{t+1} \leq [G^{(t+1)}:G^{(t)}] \delta_{G^{(t)}} + \epsilon_{G^{(t)}} + \gamma_{t}
\]
and thus that the resulting error on $G$ is no more than
\begin{equation}\label{eq:RHS-to-DRHS-error}
\sum_t [G^{(t+1)}:G^{(t)}] \delta_{G^{(t)}} + \sum_t \epsilon_{G^{(t)}}\,.
\end{equation}
As mentioned above, both the group families $\{\Z/2^n \mid n \geq 0\}$ and $\{ S_n \mid n \geq 0\}$ satisfy this subgroup chain property.\qed
\end{proof}

\paragraph{Remark.} Note that the groups $\Z/N$ for general $N$ are not treated by the results above; indeed, when $N$ is prime, there is no nontrivial tower of subgroups. (Such groups do have other relevant self-reducibility and amplification properties~\cite{HN04}.) We remark, however, that a generalization of the \textsc{Hidden Shift} problem which permits \emph{approximate equality} results in  a tight relationship between \textsc{Hidden Shift} problems for different cyclic groups. In particular, consider the \textsc{$\delta$-Approximate Hidden Shift} problem given by two functions $f, g: G \rightarrow S$ with the promise that there exists an element $s \in G$ so that $\Pr_x[g(x) = f(sx)] \geq 1 - \delta$ (where $x$ is chosen uniformly in $G$); the problem is to identify an element $s' \in G$ with this property. Note that $s'$ may not be unique in this case.

In particular, consider an instance $f, g: \Z/n \rightarrow V$ of a
\HS\ problem on a cyclic group $\Z/n$. We wish to ``lift'' this
instance to a group $\Z/m$ for $m \gg n$ in such a way that a solution
to the $\Z/m$ instance yields a solution to the $\Z/n$ instance. For a
function $\phi: \Z/n \rightarrow V$, define the function
$\hat{\phi}: \Z/m \rightarrow V$ by the rule $\hat{\phi}(x) = \phi(x \bmod n)$.
Note, then, that $\Pr_x[ \hat{f}(x) = \hat{g}(\hat{s} + x) ] \geq 1 - n/m$ for the shift $\hat{s} = s$; moreover, recovering \emph{any} shift for the $\Z/m$ problem which achieves equality with probability near $1 - n/m$ yields a solution to the $\Z/n$ problem (by taking the answer modulo $n$, perhaps after correcting for the $m \bmod n$ overhang at the end of the $\Z/m$ oracle). Note that this function is not injective.

We remark that the {\HS} problem for non-injective Boolean functions (i.e., with range $\Z/2$) sometimes admits efficient algorithms (see, e.g.,~\cite{ORR13, Roetteler16}). Whether these techniques can be extended to the general setting above is an interesting open problem.

\subsection{Selecting hard groups}\label{sec:groups}

\subsubsection{Efficiently solvable cases.}

For some choices of underlying group $G$, some of the above problems admit polynomial-time algorithms. A notable case is the \HSP~on $G = \Z$, which can be solved efficiently by Shor's algorithm~\cite{Shor94}. The \hsp~with arbitrary abelian $G$ also admits a polynomial-time algorithm~\cite{Kitaev95}. The earliest and simplest example was Simon's algorithm~\cite{Simon97}, which efficiently solves the \hsp~in the case $G = (\Z/2)^n$ and $H = \{1, s\}$ for unknown $s$, with only $O(n)$ queries to the oracle. Due to the fact that $(\Z/2)^n \rtimes \Z/2 \cong (\Z/2)^{n+1}$, Simon's algorithm also solves the \HS~problem on $(\Z/2)^n$. Additionally, \citet{FIMSS14} have given efficient (or quasi-polynomial) algorithms for hidden shifts over solvable groups of constant exponent; for example, their techniques yield efficient algorithms for the groups $(\Z/p)^n$ (for constant $p$) and $(S_4)^n$.

\subsubsection{Cyclic groups.}

In contrast with the \HSP, the general abelian \HS~is believed to be hard. The only nontrivial algorithm known is due to Kuperberg, who gave a subexponential-time algorithm for the \hsp~on dihedral groups~\cite{Kuperberg05}. He also gave a generalization to the abelian \HS~problem, as follows.
\begin{theorem} \emph{(Theorem 7.1 in~\cite{Kuperberg05})}
The abelian \HS~problem has a quantum algorithm with time and query complexity $2^{O((\log |G|)^{1/2})}$, uniformly for all finitely-generated abelian groups.
\end{theorem}
Regev and Kuperberg later improved the above algorithm (so it uses polynomial quantum space, and gains various knobs for tuning complexity parameters), but the time and query complexity remains the same~\cite{Kuperberg11, Regev04-dhsp}.

There is also evidence connecting \hsp~on the dihedral group $D_N$ (and hence also \hs~on $\Z/N$) to other hard problems. Regev showed that, if there exists an efficient quantum algorithm for the dihedral \hsp~which uses coset sampling (the only nontrivial technique known), then there's an efficient quantum algorithm for $\poly(n)$-unique-SVP~\cite{Regev04-lattice}. This problem, in turn, is the basis of several lattice-based cryptosystems. However, due to the costs incurred in the reduction, Kuperberg's algorithm only yields exponential-time attacks. An efficient solution to \hs~on $\Z/N$ could also be used to break a certain isogeny-based cryptosystem~\cite{CJS14}.

We will focus particularly on the case $\Z/{2^n}$. This is the simplest group for which all of our constructions and results apply. Moreover, basic computational tasks (encoding/decoding group elements as bitstrings, sampling uniformly random group elements, performing basic group operations, etc.) all have straightforward and extremely efficient implementations over $\Z/{2^n}$. The existence of a quantum attack with complexity $2^{O(\sqrt{n})}$ in this case will only become practically relevant in the very long term, when the costs of quantum and classical computations become somewhat comparable. If such attacks are truly a concern, then there are other natural group choices, as we discuss below.

\subsubsection{Permutation groups.}

In the search for quantum algorithms for \hsp~and \hs, arguably the most-studied group family is the family of symmetric groups $S_n$. It is well-known that an algorithm for \hsp~over $S_n \wr (\Z/2)$ would yield a polynomial-time quantum algorithm for Graph Isomorphism. As discussed in \expref{Appendix}{appendix:HSP}, this is precisely the case of \hsp~relevant to the \HS~problem over $S_n$.

For these groups, the efforts of the quantum algorithms community have so far amounted only to negative results. First, it was shown that the standard Shor-type approach of computing with individual ``coset states'' cannot succeed~\cite{MRS08}. In fact, entangled measurements over $\Omega(n \log n)$ coset states are needed~\cite{HMRRS10}, matching the information-theoretic upper bound~\cite{EHK04}. Finally, the only nontrivial technique for performing entangled measurements over multiple registers, the so-called Kuperberg sieve, is doomed to fail as well~\cite{MRS07}.

While encoding, decoding, and computing over the symmetric groups is more complicated and less efficient than the cyclic case, it is a well-understood subject (see, e.g.,~\cite{FHL80}). When discussing these groups below, we will assume (without explicit mention) an efficient solution to these problems.

\subsubsection{Matrix groups.}

Another relevant family of groups are the matrix groups $\GL_2(\F_q)$ and $\SL_2(\F_q)$ over finite fields. These nonabelian groups exhibit many structural features which are similar to the symmetric groups, such as high-dimensional irreducible representations. 
Many of the negative results concerning the symmetric groups also carry over to matrix groups~\cite{HMRRS10, MRS07}.

Efficient encoding, decoding, and computation over finite fields $\F_q$ is standard. Given these ingredients, extending to matrix groups is not complicated. In the case of $\GL_2(\F_q)$, we can encode an arbitrary pair (not both zero) $(a, c) \in \F_q^2$ in the first column, and any pair $(b, d)$ which is not a multiple of $(a, c)$ in the second column. For $\SL_2(\F_q)$, we simply have the additional constraint that $d$ is fixed to $a^{-1}(1+bc)$ by the choices of $a, b, c$.

\subsubsection{Product groups.} 

Arguably the simplest group family for which the negative results of~\cite{HMRRS10} apply, are certain $n$-fold product groups. These are groups of the form $G^n$  where $G$ is a fixed, constant-size group (e.g., $S_5$). This opens up the possibility of simply replacing the XOR operation (i.e., $\Z/2$ addition) with composition in some other constant-size group (e.g., $S_5$), and retaining the same $n$-fold product structure.

Some care is needed, however, because there do exist nontrivial algorithms in this case. When the base group $G$ is solvable, then there are efficient algorithms for both \hsp~and \HS~(see Theorem 4.17 in \cite{FIMSS14}). It is important to note that this efficient algorithm applies even to some groups (e.g., $(S_4)^n$) for which the negative results of~\cite{HMRRS10} also apply. Nevertheless, solvability seems crucial for~\cite{FIMSS14}, and choosing $G = S_5$ for the base group gives a family for which no nontrivial \HS~algorithms are known. We remark that there is however a $2^{O(\sqrt{n \log n})}$-time algorithm for order-2 {\HSP}\textsc{s} on $G^n$ based on Kuperberg's sieve~\cite{AMR07}; this suggests the possibility of subexponential (i.e., $2^{O(n^\delta)}$ for $\delta < 1$) algorithms for \HS\ over these groups.

\section{Hidden Shift Even-Mansour ciphers}\label{sec:even-mansour}

We now address the question of repairing classical symmetric-key schemes which are vulnerable to Simon's algorithm. We begin with the simplest construction, the so-called Even-Mansour cipher~\cite{EM97}.

\subsection{Generalizing the Even-Mansour scheme}

\paragraph{The standard scheme.}

The Even-Mansour construction turns a publicly known, random permutation $P : \bits^n \rightarrow \bits^n$ into a keyed, pseudorandom permutation 
\begin{align*}
E^P_{k_1, k_2} : \bits^n &\longrightarrow \bits^n\\
x &\longmapsto P(x \oplus k_1) \oplus k_2
\end{align*}
where $k_1, k_2 \in \bits^n$, and $\oplus$ denotes bitwise XOR\@. This scheme is relevant in two settings:
\begin{enumerate}
\item simply as a source of pseudorandomness; in this setting, oracle access to $P$ is provided to all parties. 
\item as a block cipher; now oracle access to both $P$ and $P^{-1}$ is provided to all parties. Access to $P^{-1}$ is required for decryption. One can then ask if $E^P$ is a PRP (adversary gets access to $E^P$), or a strong PRP (adversary gets access to both $E^P$ and its inverse).
\end{enumerate}
In all of these settings, Even-Mansour is known to be information-theoretically secure against classical adversaries making at most polynomially-many queries~\cite{EM97}.

\paragraph{Quantum chosen plaintext attacks on the standard scheme.}
The proofs of classical security of Even-Mansour carry over immediately to the setting of quantum adversaries with only classical access to the relevant oracles. However, if an adversary is granted quantum oracle access to the $P$ and $E^P$ oracles, but no access at all to the inverse oracles, then Even-Mansour is easily broken. This attack was first described in~\cite{KM12}; a complete analysis is given in~\cite{KLLN16}. The attack is simple: First, one uses the quantum oracles for $P$ and $E^P$ to create a quantum oracle for $P \oplus E^P$, i.e., the function
$$
f(x) = P(x) \oplus  P(x \oplus k_1) \oplus k_2\,.
$$
One then runs Simon's algorithm~\cite{Simon97} on the function $f$. The claim is that, with high probability, Simon's algorithm will output $k_1$. To see this, note that $f$ satisfies half of Simon's promise, namely $f(x \oplus k_1) = f(x)$. Moreover, if it is classically secure, then it \emph{almost} satisfies the entire promise. More precisely, for any fixed $P$ and random pair $(x, y)$, either the probability of a collision $f(x) = f(y)$ is low enough for Simon's algorithm to succeed, or there are so many collisions that there exists a classical attack~\cite{KLLN16}. Once we have recovered $k_1$, we also immediately recover $k_2$ with a classical query, since $k_2 = E^P_{k_1, k_2}(x) \oplus P(x \oplus k_1)$ for any $x$.

\paragraph{Hidden Shift Even-Mansour.}

To address the above attack, we propose simple variants of the Even-Mansour scheme. The construction generalizes the standard Even-Mansour scheme in the manner described in \expref{Section}{sec:intro}. Each variant is parameterized by a family of exponentially-large finite groups $G$. The general construction is straightforward to describe. We begin with a public permutation $P : G \rightarrow G$, and from it construct a family of keyed permutations 
$$
E^P_{k_1, k_2} (x) = P(x \cdot k_1) \cdot k_2\,,
$$
where $k_1, k_2$ are now uniformly random elements of $G$, and $\cdot$ denotes composition in $G$. The formal definition, as a block cipher, follows.

\begin{scheme}[Hidden Shift Even-Mansour block cipher] Let $\mathcal G$ be a family of finite, exponentially large groups, satisfying the efficient encoding conditions given in \expref{Section}{sec:groups}. The scheme consists of three polynomial-time algorithms, parameterized by a permutation $P$ of the elements of a group $G$ in $\mathcal G$:
\begin{itemize}
\item $\KeyGen : \N \rightarrow G \times G$; on input $|G|$, outputs $(k_1, k_2) \inrand G \times G$;
\item $\Enc_{k_1, k_2}^P : G \rightarrow G$; defined by $m \mapsto P(m \cdot k_1) \cdot k_2$;
\item $\Dec_{k_1, k_2}^P : G \rightarrow G$; defined by $c \mapsto P^{-1}(c \cdot k_2^{-1}) \cdot k_1^{-1}$.
\end{itemize}
\end{scheme}

For simplicity of notation, we set $E^P_{k_1, k_2} := \Enc_{k_1, k_2}^P$. Note that $\Dec_{k_1, k_2}^P = \bigl(E^P_{k_1, k_2}\bigr)^{-1}$. Correctness of the scheme is immediate; in the next section, we present several arguments for its security in various settings. All of these arguments are based on the conjectured hardness of certain \HS~problems over $\mathcal G$.

\subsection{Security reductions}

We consider two settings. In the first, the adversary is given oracle access to the permutation $P$, and then asked to distinguish the Even-Mansour cipher $E^P_{k_1, k_2}$ from a random permutation unrelated to $P$. In the second setting, the adversary is given oracle access to $P$, $P^{-1}$, as well as $E^P_{k_1, k_2}$ and its inverse; the goal in this case is to recover the key $(k_1, k_2)$ (or some part thereof).

\subsubsection{Distinguishability from random.}\label{sec:reductions-random}

We begin with the first setting described above. We fix a group $G$, and let $\mathcal P_G$ denote the family of all permutations of $G$. Select a uniformly random $P \in \mathcal P_G$. The encryption map for the Hidden Shift Even-Mansour scheme over $G$ can be written as
$$
E^P_{k_1, k_2} = L_{k_2} \circ P \circ L_{k_1}\,.
$$
If we have oracle access to $P$, then this is clearly an efficiently computable subfamily of $\mathcal P_G$, indexed by key-pairs. For pseudorandomness, the relevant problem is then to distinguish $E^P$ from a random permutation which is unrelated to the oracle $P$.

\begin{problem}[Even-Mansour Distinguishability (\emd)]
Given oracle access to permutations $P, Q \in \mathcal P_G$ and a promise that either (i.) both $P$ and $Q$ are random, or (ii.) $P$ is random and $Q = E^P_{k_1, k_2}$ for random $k_1, k_2$, decide which is the case.
\end{problem}

It is straightforward to connect this problem to the decisional version of \RHS, as follows.

\begin{proposition}\label{prop:DRHS-to-EMD}
If there exists a QPT $\algo D$ for \emd~on $\mathcal G$, then there exists a QPT algorithm for the \drhs~problem on $\mathcal G$, with soundness and completeness at most negligibly different from those of $\algo D$.
\end{proposition}
\begin{proof}
Let $f, g$ be the two oracle functions for the D\rhs~problem over $G$. We know that $f$ is a random function from $G$ to $G$, and we must decide if $g$ is also random, or simply a shift of $f$. We sample $t_1, t_2$ uniformly at random from $G$, and provide $\algo D$ with oracles $f$ (in place of $P$), and $g' := L_{t_2} \circ g \circ L_{t_1}$ (in place of $E^P$). We then simply output what $\algo D$ outputs. Note that $(f, g)$ are uniformly random permutations if and only if $(f, g')$ are. In addition, $g = f \circ L_s$ if and only if $g' = L_{t_2} \circ f \circ L_{st_1}$. It follows that the input distribution to $\algo D$ is as in EMD, modulo the fact that the oracles in \drhs~are random functions rather than random permutations. The error resulting from this is at most negligible, by the collision-finding bound of Zhandry~\cite{Zhandry15}. \qed
\end{proof}

Next, we want to amplify the \drhs~distinguisher, and then apply the reduction from \HS~given in \expref{Proposition}{prop:search-to-decision}. Combining this with \expref{Proposition}{prop:DRHS-to-EMD}, we arrive at a complete security reduction.

\begin{theorem}\label{thm:HS-to-EMD}
Let $\mathcal G$ be either the $\Z/{2^n}$ group family or the $S_n$ group family. Under \expref{Assumption}{core-assumption}, the Hidden Shift Even-Mansour cipher over $\mathcal G$ is a quantum-secure pseudorandom function. 
\end{theorem}
\begin{proof}
Let $\mathcal G$ be either the $\Z/{2^n}$ group family, or the $S_n$ group family. If the Even-Mansour cipher over $\mathcal G$ is not a qPRP, then by \expref{Definition}{def:qPRP}, there exists an algorithm $\algo D_\textsf{EMD}$ for the \textsf{EMD} problem with total (i.e., completeness plus soundness) error at most $1 - 1/s(n)$ for some polynomial $s$. To give the adversary as much freedom as possible, we assume that the probability of selecting the public permutation $P$ is taken into account here; that is, $\algo D_\textsf{EMD}$ need only succeed with inverse-polynomial probability over the choices of permutation $P$, keys $k_1, k_2$, and its internal randomness.

By \expref{Proposition}{prop:DRHS-to-EMD}, we then also have a \drhs~algorithm $\algo D_{\drhs}$ with error at most $1 - 1/s(n)$ (up to negligible terms). We can amplify this algorithm by means of a $2k$-wise independent hash function family $\mathcal H$, where $k$ is an upper bound on the running time of $\algo D_{\drhs}$ (for the given input size $n$ and required error bound $1/s(n)$). Given functions $f, g$ for the \drhs~problem on $G$, we select a random function $h \in \mathcal H$ and a random group element $t \in G$. We then call $\algo D_{\drhs}$ with oracles
$$
f' := h \circ f
\qquad \text{and} \qquad
g'_t := h \circ g \circ L_t
$$
Note that, to any efficient quantum algorithm, (i.) $f$ and $g$ are random if and only if $f'$ and $g'_{t}$ are, and (ii.) $g(x) = f(sx)$ if and only if $g'_t(x) = f'(stx)$. We know that $\algo D_{\drhs}$ will succeed with probability $1 - 1 / s(n)$, except the probability is now taken over the choice of $t$ and $h$ (rather than $f$ and $g$). We repeat this process with different random choices of $h$ and $t$. A straightforward application of a standard Chernoff bound shows that, after $O(p(n))$ runs, we will correctly distinguish with $1 - \negl(n)$ probability.

Finally, we apply \expref{Proposition}{prop:search-to-decision}, to get an algorithm for \RHS~with negligible error; by \expref{Proposition}{prop:RHS-HS-amplify}, we get an equally strong algorithm for \HS. \qed
\end{proof}

\subsubsection{Key recovery attacks.}

We now consider partial or complete key recovery attacks, in the setting where the adversary also gets oracle access to the inverses of $P$ and $E^P_{k_1, k_2}$. Note that, for the Even-Mansour cipher on any group $G$, knowing the first key $k_1$ suffices to produce the second key $k_2$, since
$$
k_2 = P(x \cdot k_1)^{-1} E^P_{k_1, k_2}(x)
$$ 
for every $x \in G$.

We remark that giving security reductions is now complicated by the fact that \RHS\ and its variants all become trivial if we are granted even a partial ability to invert $f$ or $g$; querying $f^{-1} \circ g$ on any input $x$ produces $x \cdot s^{-1}$, which immediately yields the shift $s$. However, we can still give a nontrivial reduction, as follows. 

\begin{theorem}\label{thm:HS-to-EM-key}
Consider the Even-Mansour cipher over $G \times G$, for any group $G$. Suppose there exists a QPT algorithm which, when granted oracle access to $P$, $E^{P_{k_1, k_2}}$, and their inverses, outputs $k_1, k_2$. Then there exists an efficient quantum algorithm for the \HS~problem over $G$.
\end{theorem}
\begin{proof}
We are given oracle access to functions $f, g : G \rightarrow G$ and a promise that there exists $s \in G$ such that $f(x) = g(x\cdot s)$ for all $x \in G$. We define the following oracles, which can be constructed from access to $f$ and $g$. First, we have permutations $P_f, P_g : G \times G \rightarrow G \times G$ defined by
$$
P_f (x, y) = (x, y \cdot f(x))
\qquad \text{and} \qquad
P_g(x, y) = (x, y \cdot g(x))\,.
$$
Now we sample keys $k_1 = (x_1, y_1), k_2 = (x_2, y_2)$ from $G \times G$ and define the function $E := E^{P_f}_{k_1, k_2}$. To the key-recovery adversary $\algo A$ for Even-Mansour over $G \times G$, we provide the oracles $E$ and $E^{-1}$ for the encryption/decryption oracles, and the oracles $P_g$ and $P_g^{-1}$ for the public permutation oracles.

To see that we can recover the shift $s$ from the output of $\algo A$, we rewrite $E$ in terms of $g$, as follows:
\begin{align*}
E(x, y) 
& = P_f ( x x_1, y y_1)  \cdot (x_2, y_2)\\
& = (x x_1, y y_1 f(x x_1)) \cdot (x_2, y_2)\\
& = (x x_1 s, y y_1 f(x x_1)) \cdot (s^{-1} x_2, y_2)\\
& = (x x_1 s, y y_1 g(x x_1 s)) \cdot (s^{-1} x_2, y_2)\\
& = P_g(x x_1 s, y y_1) \cdot (s^{-1} x_2, y_2)\,.
\end{align*}
After complete key recovery, $\algo A$ will output $(x_1 s, y_1)$ and $(s^{-1} x_2, y_2)$, from which we easily deduce $s$.
\qed
\end{proof}

\paragraph{Remark.}
The reduction above focuses on the problem of recovering the entire key. Note that for certain groups, e.g., $\Z/p$ for prime $p$, predicting any bit of the key with inverse-polynomial advantage is sufficient to recover the entire key (see \citet{HN04}). In such cases we may conclude that predicting individual bits of the key is difficult.

\section{Hidden Shift CBC-MACs}

\subsection{Generalizing the Encrypted-CBC-MAC scheme}

\paragraph{The standard scheme.}

The standard Encrypted-CBC-MAC construction requires a pseudorandom permutation $E_k : \bits^n \rightarrow \bits^n$. A message $m$ is subdivided into blocks $m = m_1 || m_2 || \cdots || m_l$, each of length $n$. The tag is then computed by repeatedly encrypting-and-XORing the message blocks, terminating with one additional round of encryption with a different key. Specifically, we set
$$
\text{CBC-MAC}_{k, k'} := E_{k'}( E_k ( m_l \oplus E_k ( \cdots E_k(m_2 \oplus E_k(m_1)) \cdots )))\,.
$$
This yields a secure MAC for variable-length messages.

\paragraph{Quantum chosen plaintext attacks on the standard scheme.}

If we are granted quantum CPA access to CBC-MAC$_{k, k'}$, then there is a $(\Z/2)^n$-hidden-shift attack, described below. This attack was described in~\cite{KLLN16}; another version of the attack appears in~\cite{SS17}. Consider messages consisting of two blocks, and fix the first block to be one of two distinct values $\alpha_0 \neq \alpha_1$. We use the oracle for $\text{CBC-MAC}_{k, k'}$ to construct an oracle for the function
$$
f(b, x) := \text{CBC-MAC}_{k, k'}(\alpha_b || x) = E_{k'}(E_k(x \oplus E_k(\alpha_b)))\,.
$$
Note that $f$ satisfies Simon's promise, since 
$$
f(b \oplus 1, x \oplus E_k(\alpha_0) \oplus E_k(\alpha_1)) = f(b, x)
$$ 
for all $b, x$. We can thus run Simon's algorithm to recover the string $s_k = E_k(\alpha_0) \oplus E_k(\alpha_1)$. Knowledge of $s_k$ enables us to find an exponential number of collisions, since
$$
\text{CBC-MAC}_{k, k'}(\alpha_0 || x) = \text{CBC-MAC}_{k, k'}(\alpha_1 || x \oplus E_k(\alpha_0) \oplus E_k(\alpha_1))\,.  
$$
In particular, this CBC-MAC does not satisfy the Boneh-Zhandry notion of a secure MAC in the quantum world~\cite{BZ13a}.

\paragraph{Hidden Shift CBC-MAC.}

We propose generalizing the Encrypted-CBC-MAC construction above, by allowing the bitwise XOR operation to be replaced by composition in some exponentially-large family of finite groups $G$. Each message block is then identified with an element of $G$, and we view the pseudorandom permutation $E_k$ as a permutation of the group elements of $G$. We then define
\begin{align*}
\text{CBC-MAC}^G_{k, k'} : G^* &\longrightarrow G \\
(m_1, \dots, m_l) &\longmapsto E_{k'}( E_k ( m_l \cdot E_k ( \cdots E_k(m_2 \cdot E_k(m_1)) \cdots )))\,,
\end{align*}
where $\cdot$ denotes the group operation in $G$.\\

\begin{scheme}[Hidden Shift Encrypted-CBC-MAC]
Let $G$ be a family of finite, exponentially large groups satisfying the efficient encoding conditions given in \expref{Section}{sec:groups}. Let $E_k : G \rightarrow G$ be a quantum-secure pseudorandom permutation. The scheme consists of three polynomial-time algorithms:
\begin{itemize}
\item $\KeyGen$; on input $|G|$, outputs two keys $k, k'$ using key generation for $E$;
\item $\Mac_{k, k'} : m \longmapsto E_{k'}( E_k ( m_l \cdot E_k ( \cdots E_k(m_2 \cdot E_k(m_1)) \cdots )$;
\item $\Ver_{k, k'} : (m, t) \mapsto \textsf{accept}$ if $\Mac_{k, k'}(m) = t$, and \textsf{reject} otherwise.
\end{itemize}
\end{scheme}

We consider the security of this scheme in the next section.

\subsection{Security reduction}

We now give a reduction from the \RHS~problem to collision-finding in the above CBC-MAC.

\begin{theorem}
Let $\mathcal G$ be either the $\Z/{2^n}$ group family or the $S_n$ group family. Under \expref{Assumption}{core-assumption}, the Hidden-Shift CBC-MAC over $\mathcal G$ is a collision-resistant function.
\end{theorem}
\begin{proof}
For simplicity, we assume that the collision-finding adversary finds collisions between equal-length messages. This is of course trivially true, for example, if the MAC is used only for messages of some a priori fixed length.

Suppose we are given an instance of the \textsc{Hidden Shift} problem, i.e., a pair of functions $F_0, F_1$ with the promise that $F_0$ is random and $F_1$ is a shift of $F_0$. We have at our disposal a QPT $\algo A$ which finds collisions in the Hidden Shift Encrypted-CBC-MAC\@. We assume without loss of generality that, whenever $\algo A$ outputs a collision $(c, c')$, there is no pair of prefixes of $(c, c')$ that also give a valid collision; indeed, we can easily build an $\algo A'$ which, whenever such prefixes exist, simply outputs the prefix collision instead. 

We assume for the moment that the number of message blocks in $c$ and $c'$ is the same number $t$. Since the number of blocks and the running time of $\algo A$ are polynomial, we can simply guess $t$, and we will guess correctly with inverse-polynomial probability. We run $\algo A$ with a modified oracle $\mathcal O$ which ``inserts'' our hidden shift problem at stage $t$. This is defined as follows. 

Let $m$ be our input message, and $l$ the number of blocks. If $l < t$, we simply output the usual Encrypted-CBC-MAC of $m$. If $l \geq t$, we first perform $t-1$ rounds of the CBC procedure, computing a function
$$
h (m) :=  E_k( m_{t-1} \cdot E_k(\cdots E_k(m_2 \cdot E_k(m_1)) \cdots )\,.
$$
Note that $h$ only depends on the first $t-1$ blocks of $m$. Next, we choose a random bit $b$ 
and compute $F_{b(m)}(m_t \cdot h(m))$. We then finish the rest of the rounds of the CBC procedure, outputting
$$
O(m) := E_k' ( E_k ( m_l \cdot E_k(\cdots E_k(F_{b(m)}(m_t \cdot h(m))) \cdots )\,.
$$
It's not hard to see that the distribution that the adversary observes will be indistinguishable from the usual Encrypted-CBC-MAC\@. Suppose a collision $(m, m')$ is output. We set $x_1 = m_1 || m_2 || \cdots || m_{t-1}$ and $x_2 = m_1' || m_2' || \cdots || m_{t-1}'$ and $y_1 = m_t$ and $y_2 = m_t'$. The collision then means that
$$
F_{b(m)}(y_1 \cdot h(x_1)) = F_{b(m')}(y_2 \cdot h(x_2))\,.
$$
Since $m \neq m'$, with probability $1/2$ we have $b(m) \neq b(m')$. We repeat $\algo A$ until we achieve inequality of these bits. We then have
$$
F_0(y_1 \cdot h(x_1)) = F_1(y_2 \cdot h(x_2)) = F_0(y_2 \cdot h(x_2) \cdot s)
$$
and so the shift is simply $s = y_2^{-1} h(x_2)^{-1} y_1 h(x_1)$.
\qed
\end{proof}

\section{Thwarting the Simon attack on other schemes}

It is reasonable to conjecture that our transformation secures (classically secure) symmetric-key schemes against quantum CPA, generically. So far, we have only been able to give complete security reductions in the cases of the Even-Mansour cipher and the Encrypted-CBC-MAC\@. For the case of all other schemes vulnerable to the Simon algorithm attacks of~\cite{KM10, KM12, KLLN16}, we can only say that the attack is thwarted by passing from $(\Z/2)^n$ to $\Z/{2^n}$ or $S_n$. We now briefly outline two cases of particular note. For further details, see \expref{Appendix}{sec:appendix-other}.

The first case is the Feistel network construction, which transforms random functions into pseudorandom permutations. While the three-round Feistel cipher is known to be classically secure~\cite{LR88}, no security proof is known in the quantum CPA case, for any number of rounds. In~\cite{KM10}, a quantum chosen-plaintext attack is given for the three-round Feistel cipher, again based on Simon's algorithm. The attack is based on the observation that, if one fixes the first half of the input to one of two fixed values $\alpha_0 \neq \alpha_1$, then the output contains one of two functions $f_{\alpha_0}$, $f_{\alpha_1}$, which are $(\Z/2)^n$-shifts of each other. However, if we instead replace each bitwise XOR in the Feistel construction with addition modulo $\Z/{2^n}$, the two functions become $\Z/{2^n}$-shifts, and the attack now requires a cyclic \HS~subroutine.

The second case is what~\cite{KLLN16} refer to as the ``quantum slide attack,'' which uses Simon's algorithm to give a linear-time quantum chosen-plaintext attack, an exponential speedup over classical slide attacks. The attack works against ciphers $E_{k, t}(x) := k \oplus (R_k)^t(x)$ which consist of $t$ rounds of a function $R_k(x) := R(x \oplus k)$. In the attack, one simply observes that $E_{k, t} (R(x))$ is a shift of $R(E_{k, t} (x))$ by the key $k$, and then applies Simon's algorithm. To defeat this attack, we simply work over $\Z/{2^n}$, setting $E_{k, t}(x) := k + (R_k)^t(x)$ and $R_k(x) := R(x+k)$. It's easy to see that the same attack now requires a \HS~subroutine for $\Z/{2^n}$.

\section*{Acknowledgments}

G.\,A.\ would like to thank Tommaso Gagliardoni and Christian Majenz for helpful discussions. G.A. acknowledges financial support from the European Research Council (ERC Grant Agreement 337603), the Danish Council for Independent Research (Sapere Aude) and VILLUM FONDEN via the QMATH Centre of Excellence (Grant 10059). A.\,R.\ acknowledges support from NSF grant IIS-1407205.

\bibliography{alagic-full-bib}

\begin{thebibliography}{35}
\providecommand{\natexlab}[1]{#1}
\providecommand{\url}[1]{\texttt{#1}}
\expandafter\ifx\csname urlstyle\endcsname\relax
  \providecommand{\doi}[1]{doi: #1}\else
  \providecommand{\doi}{doi: \begingroup \urlstyle{rm}\Url}\fi

\bibitem[Alagic et~al.(2007)Alagic, Moore, and Russell]{AMR07}
Gorjan Alagic, Cristopher Moore, and Alexander Russell.
\newblock Quantum algorithms for {Simon's} problem over general groups.
\newblock In \emph{Proceedings of the Eighteenth Annual {ACM-SIAM} Symposium on
  Discrete Algorithms (SODA)}, pages 1217--1224. ACM Press, 2007.

\bibitem[Boneh and Zhandry(2013)]{BZ13a}
Dan Boneh and Mark Zhandry.
\newblock Quantum-secure message authentication codes.
\newblock In \emph{32nd Annual International Conference on the Theory and
  Applications of Cryptographic Techniques, Advances in Cryptology -
  {EUROCRYPT} 2013}, pages 592--608. Springer, 2013.
\newblock \doi{10.1007/978-3-642-38348-9_35}.

\bibitem[Chen et~al.(2016)Chen, Jordan, Liu, Moody, Peralta, Perlner, and
  Smith-Tone]{NIST16}
Lily Chen, Stephen Jordan, Yi-Kai Liu, Dustin Moody, Rene Peralta, Ray Perlner,
  and Daniel Smith-Tone.
\newblock Report on post-quantum cryptography.
\newblock Technical report, National Institute of Standards and Technology,
  2016.
\newblock URL \url{http://nvlpubs.nist.gov/nistpubs/ir/2016/NIST.IR.8105.pdf}.

\bibitem[Childs et~al.(2014)Childs, Jao, and Soukharev]{CJS14}
Andrew~M. Childs, David Jao, and Vladimir Soukharev.
\newblock Constructing elliptic curve isogenies in quantum subexponential time.
\newblock \emph{J. Mathematical Cryptology}, 8\penalty0 (1):\penalty0 1--29,
  2014.
\newblock \doi{10.1515/jmc-2012-0016}.

\bibitem[Dinh et~al.(2015)Dinh, Moore, and Russell]{DMR15}
Hang Dinh, Cristopher Moore, and Alexander Russell.
\newblock Limitations of single coset states and quantum algorithms for {Code
  Equivalence}.
\newblock \emph{Quantum Info. Comput.}, 15\penalty0 (3-4):\penalty0 260--294,
  March 2015.

\bibitem[Dunkelman et~al.(2015)Dunkelman, Keller, and Shamir]{DKS15}
Orr Dunkelman, Nathan Keller, and Adi Shamir.
\newblock Slidex attacks on the {Even--Mansour} encryption scheme.
\newblock \emph{Journal of Cryptology}, 28\penalty0 (1):\penalty0 1--28, 2015.
\newblock ISSN 1432-1378.
\newblock \doi{10.1007/s00145-013-9164-7}.

\bibitem[Ettinger et~al.(2004)Ettinger, H{\o}yer, and Knill]{EHK04}
Mark Ettinger, Peter H{\o}yer, and Emanuel Knill.
\newblock The quantum query complexity of the hidden subgroup problem is
  polynomial.
\newblock \emph{Information Processing Letters}, 91\penalty0 (1):\penalty0
  43--48, 2004.
\newblock \doi{10.1016/j.ipl.2004.01.024}.

\bibitem[Even and Mansour(1997)]{EM97}
Shimon Even and Yishay Mansour.
\newblock A construction of a cipher from a single pseudorandom permutation.
\newblock \emph{Journal of Cryptology}, 10\penalty0 (3):\penalty0 151--161,
  1997.
\newblock \doi{10.1007/s001459900025}.

\bibitem[Fenner and Zhang(2013)]{FZ13}
Stephen Fenner and Yong Zhang.
\newblock On the complexity of the hidden subgroup problem.
\newblock \emph{International Journal of Foundations of Computer Science},
  24\penalty0 (8):\penalty0 1221--1234, 2013.

\bibitem[Friedl et~al.(2014)Friedl, Ivanyos, Magniez, Santha, and Sen]{FIMSS14}
Katalin Friedl, G{\'{a}}bor Ivanyos, Fr{\'{e}}d{\'{e}}ric Magniez, Miklos
  Santha, and Pranab Sen.
\newblock Hidden translation and translating coset in quantum computing.
\newblock \emph{{SIAM} J. Comput.}, 43\penalty0 (1):\penalty0 1--24, 2014.
\newblock \doi{10.1137/130907203}.

\bibitem[Furst et~al.(1980)Furst, Hopcroft, and Luks]{FHL80}
Merrick Furst, John Hopcroft, and Eugene Luks.
\newblock Polynomial-time algorithms for permutation groups.
\newblock In \emph{Proceedings of the 21st Annual Symposium on Foundations of
  Computer Science}, FOCS '80, pages 36--41, Washington, DC, USA, 1980. IEEE
  Computer Society.
\newblock \doi{10.1109/SFCS.1980.34}.

\bibitem[Hallgren et~al.(2010)Hallgren, Moore, R\"{o}tteler, Russell, and
  Sen]{HMRRS10}
Sean Hallgren, Cristopher Moore, Martin R\"{o}tteler, Alexander Russell, and
  Pranab Sen.
\newblock Limitations of quantum coset states for graph isomorphism.
\newblock \emph{J. ACM}, 57\penalty0 (6):\penalty0 34:1--34:33, November 2010.
\newblock \doi{10.1145/1857914.1857918}.

\bibitem[H{\aa}stad and N{\aa}slund(2004)]{HN04}
Johan H{\aa}stad and Mats N{\aa}slund.
\newblock The security of all {RSA} and {Discrete Log} bits.
\newblock \emph{J. ACM}, 51\penalty0 (2):\penalty0 187--230, March 2004.
\newblock \doi{10.1145/972639.972642}.

\bibitem[Kaplan et~al.(2016)Kaplan, Leurent, Leverrier, and
  Naya-Plasencia]{KLLN16}
Marc Kaplan, Ga\"etan Leurent, Anthony Leverrier, and Mar\'ia Naya-Plasencia.
\newblock Breaking symmetric cryptosystems using quantum period finding.
\newblock In \emph{Advances in Cryptology -- CRYPTO 2016, 36th Annual
  International Cryptology Conference, Santa Barbara, CA, USA, August 14-18,
  2016}, pages 207--237. Springer, 2016.
\newblock \doi{10.1007/978-3-662-53008-5_8}.

\bibitem[Kitaev(1995)]{Kitaev95}
Alexei~Y. Kitaev.
\newblock Quantum measurements and the abelian stabilizer problem.
\newblock Technical Report quant-ph/9511026, arXiv, November 1995.

\bibitem[Kuperberg(2005)]{Kuperberg05}
Greg Kuperberg.
\newblock A subexponential-time quantum algorithm for the dihedral hidden
  subgroup problem.
\newblock \emph{SIAM J. Comput.}, 35\penalty0 (1):\penalty0 170--188, July
  2005.
\newblock \doi{10.1137/S0097539703436345}.

\bibitem[Kuperberg(2011)]{Kuperberg11}
Greg Kuperberg.
\newblock Another subexponential-time quantum algorithm for the dihedral hidden
  subgroup problem.
\newblock Technical Report quant-ph/1112.3333, arXiv, December 2011.

\bibitem[Kuwakado and Morii(2010)]{KM10}
Hidenori Kuwakado and Masakatu Morii.
\newblock Quantum distinguisher between the 3-round {Feistel} cipher and the
  random permutation.
\newblock In \emph{Information Theory Proceedings (ISIT), 2010 IEEE
  International Symposium on}, pages 2682--2685, June 2010.
\newblock \doi{10.1109/ISIT.2010.5513654}.

\bibitem[Kuwakado and Morii(2012)]{KM12}
Hidenori Kuwakado and Masakatu Morii.
\newblock Security on the quantum-type {Even-Mansour} cipher.
\newblock In \emph{Proceedings of the International Symposium on Information
  Theory and its Applications (ISITA)}, pages 312--316. {IEEE} Computer
  Society, 2012.

\bibitem[Luby and Rackoff(1988)]{LR88}
Michael Luby and Charles Rackoff.
\newblock How to construct pseudorandom permutations from pseudorandom
  functions.
\newblock \emph{SIAM J. Comput.}, 17(2):\penalty0 337--386, 1988.

\bibitem[Moore et~al.(2007)Moore, Russell, and {\'{S}niady}]{MRS07}
Cristopher Moore, Alexander Russell, and Piotr {\'{S}niady}.
\newblock On the impossibility of a quantum sieve algorithm for graph
  isomorphism.
\newblock In \emph{Proceedings of the Thirty-ninth Annual ACM Symposium on
  Theory of Computing}, STOC '07, pages 536--545, New York, NY, USA, 2007. ACM.
\newblock \doi{10.1145/1250790.1250868}.

\bibitem[Moore et~al.(2008)Moore, Russell, and Schulman]{MRS08}
Cristopher Moore, Alexander Russell, and Leonard~J. Schulman.
\newblock The symmetric group defies strong fourier sampling.
\newblock \emph{SIAM J. Comput.}, 37:\penalty0 1842--1864, March 2008.
\newblock \doi{10.1137/050644896}.

\bibitem[Ozols et~al.(2013)Ozols, Roetteler, and Roland]{ORR13}
M\=aris Ozols, Martin Roetteler, and J{\'e}r{\'e}mie Roland.
\newblock Quantum rejection sampling.
\newblock \emph{ACM Trans. Comput. Theory}, 5\penalty0 (3):\penalty0
  11:1--11:33, August 2013.
\newblock ISSN 1942-3454.
\newblock \doi{10.1145/2493252.2493256}.

\bibitem[Patel et~al.(2003)Patel, Ramzan, and Sundaram]{PRS03}
Sarvar Patel, Zulfikar Ramzan, and Ganpathy~S. Sundaram.
\newblock {Luby-Rackoff} ciphers: Why {XOR} is not so exclusive.
\newblock In \emph{Revised Papers from the 9th Annual International Workshop on
  Selected Areas in Cryptography}, SAC '02, pages 271--290, London, UK, UK,
  2003. Springer-Verlag.
\newblock ISBN 3-540-00622-2.

\bibitem[Regev(2004{\natexlab{a}})]{Regev04-dhsp}
Oded Regev.
\newblock A subexponential time algorithm for the dihedral hidden subgroup
  problem with polynomial space.
\newblock Technical Report quant-ph/0406151, arXiv, June 2004{\natexlab{a}}.

\bibitem[Regev(2004{\natexlab{b}})]{Regev04-lattice}
Oded Regev.
\newblock Quantum computation and lattice problems.
\newblock \emph{SIAM J. Comput.}, 33\penalty0 (3):\penalty0 738--760, March
  2004{\natexlab{b}}.
\newblock \doi{10.1137/S0097539703440678}.

\bibitem[Regev(2005)]{Regev05}
Oded Regev.
\newblock On lattices, learning with errors, random linear codes, and
  cryptography.
\newblock In \emph{Proceedings of the Thirty-seventh Annual ACM Symposium on
  Theory of Computing}, STOC '05, pages 84--93, New York, NY, USA, 2005. ACM.
\newblock \doi{10.1145/1060590.1060603}.

\bibitem[Roetteler(2016)]{Roetteler16}
Martin Roetteler.
\newblock Quantum algorithms for abelian difference sets and applications to
  dihedral hidden subgroups.
\newblock In \emph{11th Conference on the Theory of Quantum Computation,
  Communication and Cryptography, {TQC} 2016, September 27-29, Berlin.}, 2016.

\bibitem[Santoli and Schaffner(2017)]{SS17}
Thomas Santoli and Christian Schaffner.
\newblock Using {Simon's} algorithm to attack symmetric-key cryptographic
  primitives.
\newblock \emph{Quantum Information {\&} Computation}, 17\penalty0
  (1{\&}2):\penalty0 65--78, 2017.

\bibitem[Shor(1994)]{Shor94}
Peter~W. Shor.
\newblock Algorithms for quantum computation: {discrete} logarithms and
  factoring.
\newblock In \emph{Proceedings of the 35th Annual Symposium on Foundations of
  Computer Science}, FOCS '94, pages 124--134, Washington, DC, USA, 1994. IEEE
  Computer Society.
\newblock ISBN 0-8186-6580-7.
\newblock \doi{10.1109/SFCS.1994.365700}.

\bibitem[Simon(1997)]{Simon97}
Daniel~R. Simon.
\newblock On the power of quantum computation.
\newblock \emph{SIAM J. Comput.}, 26\penalty0 (5):\penalty0 1474--1483, October
  1997.
\newblock \doi{10.1137/S0097539796298637}.

\bibitem[{Zhandry}(2016)]{Zhandry16}
M.~{Zhandry}.
\newblock A note on quantum-secure {PRPs}.
\newblock \emph{arXiv preprint arXiv:1607.07759}, November 2016.

\bibitem[Zhandry(2012{\natexlab{a}})]{Zhandry12}
Mark Zhandry.
\newblock How to construct quantum random functions.
\newblock In \emph{Proceedings of the 2012 IEEE 53rd Annual Symposium on
  Foundations of Computer Science}, FOCS '12, pages 679--687, Washington, DC,
  USA, 2012{\natexlab{a}}. IEEE Computer Society.
\newblock ISBN 978-0-7695-4874-6.
\newblock \doi{10.1109/FOCS.2012.37}.

\bibitem[Zhandry(2012{\natexlab{b}})]{Zhandry12a}
Mark Zhandry.
\newblock Secure identity-based encryption in the quantum random oracle model.
\newblock In \emph{Advances in Cryptology -- CRYPTO 2012, 32nd Annual
  International Cryptology Conference, Santa Barbara, CA, USA, August 19-23,
  2012}, pages 758--775. Springer, 2012{\natexlab{b}}.
\newblock ISBN 978-3-642-32008-8.
\newblock \doi{10.1007/978-3-642-32009-5_44}.

\bibitem[Zhandry(2015)]{Zhandry15}
Mark Zhandry.
\newblock A note on the quantum collision and set equality problems.
\newblock \emph{Quantum Info. Comput.}, 15\penalty0 (7-8):\penalty0 557--567,
  May 2015.

\end{thebibliography}

\appendix

\section{Hidden Subgroups and Hidden Shifts}\label{appendix:HSP} 

\subsubsection{Basic definitions.}
We now briefly discuss the \HSP, which is closely related to \HS.
\begin{problem}[\HSP~(\hsp)]
Let $G$ be a group and $S$ a set. Given a function $f : G \rightarrow S$, and a promise that there exists $H \leq G$ such that $f$ is constant and distinct on the right cosets of $H$, output a complete set of generators for $H$.
\end{problem}

Another, equivalent formulation of the \hsp~promise on $f$ is that for $x \neq y$, $f(x) = f(y)$ iff $x = h \cdot y$ for $h \in H$. As before, one can also consider decision versions of \hsp~(e.g., where one has to decide if $f$ hides a trivial or nontrivial subgroup) and promise versions where the function $f$ is a random function satisfying the constraint that $f(hx) = x$ for all $x \in G$ and $h \in H$. This last variant is important for our purposes so we separately define it.

\begin{problem}[\RHSP~(\rhsp)]
Let $G$ be a group and $S$ a set. Given a function $f : G \rightarrow S$ chosen uniformly among all functions for which $f(x) = f(hx)$ for all $x \in G$ and $h \in H$, output a complete set of generators for $H$.
\end{problem}

\subsubsection{Some reductions.}

Traditionally, the Hidden Subgroup problem (HSP) has played a prominent role in the literature, as it offers a simple framework to which many other problems can be directly reduced. Indeed, there is a general reduction from \hs~to \hsp.

\paragraph{A canonical reduction from \hs~to \hsp.} Consider an instance
of a \HS~problem over $G$ given by the functions
$f_0, f_1: G \rightarrow S$ such that $f_0(x) = f_1(x \cdot
s)$. Recall that the \emph{wreath product} $K = G \wr \Z/2$ is the
semi-direct product $(G \times G) \rtimes \Z/2$, where the action of
the nontrivial element of $\Z/2$ on $G \times G$ is the swap
$(a, b) \mapsto (b, a)$. Now define the function
\begin{align*}
\varphi : (G \times G) \rtimes \Z/2 &\longrightarrow S \times S\\
((x, y), b) &\longmapsto (f_b(x), f_{b \oplus 1}(y))\,.
\end{align*}
One then easily checks that the function $\varphi$ is constant and distinct on the cosets of the order-two subgroup of $K$ generated by $((s, s^{-1}), 1)$. 

We remark that the reduction above can significantly ``complicate'' the underlying group. In particular, note that $A \rtimes \Z/2$ is always non-abelian (unless the action of $1 \in \Z/2 = \{0,1\}$ on $A$ is trivial). Note that this reduction does not yield a reduction from \rhs~to \rhsp~as the resulting \hsp~instance is not uniformly random in the fully random case.

\paragraph{The special case of $(\Z/2)^n$; reductions from \rhs~to \rhsp.}
On $(\Z/2)^n$, the \HS~problem can be reduced to the \RHSP~on $(\Z/2)^n$ via a special reduction that exploits
$\Z/2$ structure. Specifically, for a pair of injective functions
$f_0, f_1: (\Z/2)^n \rightarrow S$ (for which
$f_0(x) = f_1(x \oplus s)$,) construct the function
$f: (\Z/2)^{n+1} \rightarrow S$ so that
\[
g(bx) = f_b(x)\,,\qquad\text{for $b \in \Z/2$ and $x \in (\Z/2)^n$.}
\]
Then observe that $g$ hides the subgroup generated by $1s$.

Note, furthermore, that if the $f_i$ are (independent) random
functions, then the same can be said of $g$; likewise, if $f_1$ is a
shift of the random function $f_0$, the function $g$ is precisely a
random function subject to the constraint that
$g(x) = g(x \oplus 1s)$; thus this reduces \rhs~to \rhsp. In this \rhs~setting, it is possible to develop an alternate reduction that more closely resembles the attacks we discussed above. Specifically, given the functions $f_0,f_1: (\Z/2)^n \rightarrow (\Z/2)^n$ (so that $S = (\Z/2)^n$), consider the oracle $g = f_0 \oplus f_1$. When $f_1(x) = f_0(x \oplus s)$, note that this oracle satisfies the symmetry condition
\begin{align*}
[f_0 \oplus f_1](x \oplus  s) &= f_0(x \oplus  s) \oplus  f_1(x \oplus  s)\\
  &= f_1(x \oplus  s \oplus  s) \oplus f_0(x) = f_0(x) \oplus f_1(x) = [f_0 \oplus f_1] (x)
\end{align*}
so that $g = f_0 \oplus f_1$ is a random function subject to the constraint that $g(x) = g(x \oplus s)$, as desired. If the functions $f_i$ are independent, the function $g$ has the uniform distribution. Thus this reduces \rhs~to \rhsp.

\section{Other Hidden Shift Constructions}\label{sec:appendix-other}

\subsection{Feistel ciphers}

\paragraph{The standard scheme.}

The Feistel cipher is a method for turning random functions into pseudorandom permutations. The core ingredient is a one-round Feistel cipher, which, for a function $f: \{0,1\}^n \rightarrow \{0,1\}^n$, is given by
\begin{align*}
F_f : \{0,1\}^{2n} &\longrightarrow \{0,1\}^{2n}\\
x || y &\longmapsto y \oplus f(x) || x\,.
\end{align*}
The function $f$ is called the ``round function.'' The multi-round version of the Feistel cipher is defined by concatenating multiple one-round ciphers, each with a different choice of round function. Of particular interest is the three-round cipher, defined by
$$
F_{R_1, R_2, R_3}(x || y) = F_{R_3}(F_{R_2}(F_{R_1}(x || y)))\,.
$$
A well-known result of Luby and Rackoff says that, if the $R_j$ are random and independent, then $F_{R_1, R_2, R_3}$ is indistinguishable from a random permutation~\cite{LR88}.

\paragraph{Quantum chosen plaintext attacks on the standard scheme.}

Suppose we are given quantum oracle access to a function $F$, and promised that $F$ is either a random permutation, or that $F := F_{R_1, R_2, R_3}$ for some unknown, random functions $R_j$. The following attack was first shown in~\cite{KM10}; a thorough analysis appears in~\cite{KLLN16}. We first fix two $n$-bit strings $\alpha_0 \neq \alpha_1$. We then use the oracle for $F$ to build oracles $f_0, f_1$ defined by
$$
f_b(y) := F(\alpha_b || y)\bigr|_{n+1}^{2n}  \oplus \alpha_b\,.
$$
Here $s|_{j}^k := s_j s_{j+1} \cdots s_k$. We then run Simon's algorithm to see if there's a shift between $f_0$ and $f_1$. If a shift is produced, we output ``Feistel.'' Otherwise we output ``random.'' 

To see why the attack is sucessful, first note that if $F$ is a random permutation, then the $f_b$ are random functions. On the other hand, if $F = F_{R_1, R_2, R_3}$, then one easily checks that $f_b(y) = R_2(y \oplus R_1(\alpha_b))$. We then have
$$
f_1(y) = f_0(y \oplus (R_1(\alpha_0) \oplus R_1(\alpha_1)))
$$
for all $y$. Since the $R_j$ are random, one can check that there are not too many other collisions~\cite{KLLN16}. It follows that Simon's algorithm will output $R_1(\alpha_0) \oplus R_1(\alpha_1)$ with high probability.

\paragraph{Hidden Shift Feistel cipher.}

For simplicity, we will work over the group $\Z/{2^n}$. Our construction generalizes to other group families in a straightforward way. Given a function $f : \Z/{2^n} \rightarrow \Z/{2^n}$, we define the one-round Feistel cipher using round function $f$ to be
\begin{align*}
F_f : \Z/{2^n} \times  \Z/{2^n}  &\longrightarrow \Z/{2^n}  \times \Z/{2^n}\\
(x, y) &\longmapsto (y + f(x), x)\,.
\end{align*}
Since $\Z/{2^n} \times \Z/{2^n} \cong \Z/{2^{2n}}$, we can then view $F_f$ as a permutation on $2n$-bit strings. 
For multi-round ciphers, we define composition as before. In particular, given three functions $R_1, R_2, R_3$ on $G$, we get the three-round Feistel cipher
\begin{align*}
F_{R_1, R_2, R_3} : \Z/{2^n} \times \Z/{2^n} &\longrightarrow \Z/{2^n} \times \Z/{2^n}\\
(x, y) &\longmapsto F_{R_3}(F_{R_2}(F_{R_1}(x, y)))\,.
\end{align*}

Next, we check that the attack of~\cite{KM10} now appears to require a subroutine for the Hidden Shift problem over $\Z/{2^n}$, contrary to our Cyclic Hidden Shift Assumption (\expref{Assumption}{cyclic-assumption}). We are given a function $F$ on $\Z/{2^{2n}}$ with the promise that $F$ is either random or a Feistel cipher $F_{R_1, R_2, R_3}$ as above. Proceeding precisely as before, we pick two elements $\alpha_0 \neq \alpha_1$ of $\Z/{2^n}$ and build two functions
$$
f_b(y) := F(\alpha_b, y)|_{n+1}^{2n} - \alpha_b\,.
$$
If $F$ is random, then clearly so are $f_0$ and $f_1$. But if $F = F_{R_1, R_2, R_3}$ then one easily checks that
$$
f_b(y) = R_2(y + R_1(\alpha_b))\,,
$$
from which it follows that
$$
f_1(y) = R_2(y + R_1(\alpha_1)) = R_2(y + R_1(\alpha_0) - R_1(\alpha_0) + R_1(\alpha_1))  = f_0(y + s)
$$
where we set $s := R_1(\alpha_1) - R_1(\alpha_0)$. We are thus presented with a \HS~problem over the group $\Z/{2^n}$, which is hard according to \expref{Assumption}{cyclic-assumption}. The only known subroutine (analogous to Simon) that one could apply here would be Kuperberg's algorithm, which would find $s$ in time $2^{\Theta(\sqrt{n})}$~\cite{Kuperberg05}. Defining the Feistel network over other groups (such as $S_n$) would frustrate all nontrivial quantum-algorithmic approaches, including the Kuperberg approach~\cite{MRS07}.

\subsection{Protecting against quantum slide attacks}

\paragraph{Quantum slide attack.}

Classically, slide attacks are a class of subexponential-time attacks against ciphers which encrypt simply by repeatedly applying some function $R_k$, with a fixed key $k$. Kaplan et al.~\cite{KLLN16} showed how Simon's algorithm can be used to give a polynomial-time ``quantum slide attack'' against ciphers of the form
$$
E_{k, t} := k \oplus R_k^t(x) = k \oplus (R_k \circ R_k \cdots \circ R_k)(x)\,,
$$
where $R_k(x) = R(x \oplus k)$, and $R$ is a known permutation. As usual, the attack requires quantum CPA access to $E_{k, t}$. The attack follows directly from the observation that the functions
$$
f_b(x) =
\begin{cases}
E_{k, t} (R(x)) \oplus x &\text{if $b = 0$,}\\
R(E_{k, t} (x)) \oplus x &\text{if $b = 1$.}
\end{cases}
$$
are shifts of each other by the key $k$, i.e., $f_0(x \oplus k)  = f_1(x)$ for all $x$. This means we can extract $k$ with Simon's algorithm.

\paragraph{Eliminating the attack via hidden shifts.}

Following our established pattern, we adapt schemes $E_{k, t}$ to use modular addition over $\Z/{2^n}$ instead of bitwise XOR\@. Given a permutation $R$ of $\bits^n$, we now set
$$
R_k(x) := R(x + k)
\qquad \text{and} \qquad
E_{k, t} := k + R_k^t(x)\,.
$$
Proceeding with the attack as before, we now define
$$
f_b(x) =
\begin{cases}
E_{k, t} (R(x)) - x &\text{if $b = 0$,}\\
R(E_{k, t} (x)) - x &\text{if $b = 1$.}
\end{cases}
$$
We then check that
\begin{align*}
f_0(x + k) 
&= E_{k, t}(R(x + k)) - (x + k) \\
&= k + R_k^t(R(x + k)) - x - k \\
&= R(E_{k, t}(x)) - x\\
&= f_1(x)\,.
\end{align*}
Continuing as in the Simon attack would now require a solution to the \HS~problem over $\Z/{2^n}$.

\end{document}